\documentclass[a4paper,UKenglish,cleveref, autoref, thm-restate]{lipics-v2019}

\usepackage{dirtytalk}

\newtheorem{observation}[theorem]{Observation}

\newcommand{\cb}{\mathcal{B}}
\newcommand{\co}{\mathcal{O}}
\newcommand{\ca}{\mathcal{A}}

\newcommand{\ch}{\mathcal{H}}

\newcommand{\cl}{\mathcal{L}}

\newcommand{\cost}{\mathrm{cost}}
\newcommand{\sol}{\sigma}
\newcommand{\solbar}{\overline{\sigma}}

\title{Improved Bounds for Metric Capacitated Covering Problems}

\titlerunning{Improved Bounds for Metric Capacitated Covering Problems} 

\author{Sayan Bandyapadhyay}{Department of Informatics, University of Bergen, Norway}{sayan.bandyapadhyay@gmail.com}{https://orcid.org/0000-0001-8875-0102}{}

\authorrunning{S. Bandyapadhyay} 

\Copyright{Sayan Bandyapadhyay} 



\category{} 

\relatedversion{} 

\supplement{}

\funding{This work is partly supported by the Research Council of Norway via the project ``MULTIVAL''.}

\acknowledgements{I am indebted to Tanmay Inamdar for giving invaluable feedback on this work. I also thank the anonymous reviewers whose suggestions have helped to further improve the quality of the paper.}

\nolinenumbers 

\begin{document}
\maketitle

\begin{abstract}
In the Metric Capacitated Covering (MCC) problem, given a set of balls $\cb$ in a metric space $P$ with metric $d$ and a capacity parameter $U$, the goal is to find a minimum sized subset ${\cb}'\subseteq \cb$ and an assignment of the points in $P$ to the balls in $\cb'$ such that each point is assigned to a ball that contains it and each ball is assigned with at most $U$ points. MCC achieves an $O(\log |P|)$-approximation using a greedy algorithm. On the other hand, it is hard to approximate within a factor of $o(\log |P|)$ even with $\beta < 3$ factor expansion of the balls. Bandyapadhyay~{et al.} [SoCG 2018, DCG 2019] showed that one can obtain an $O(1)$-approximation for the problem with $6.47$ factor expansion of the balls. An open question left by their work is to reduce the gap between the lower bound $3$ and the upper bound $6.47$. In this current work, we show that it is possible to obtain an $O(1)$-approximation with only $4.24$ factor expansion of the balls. We also show a similar upper bound of $5$ for a more generalized version of MCC for which the best previously known bound was $9$.           
\end{abstract}

\section{Introduction}
In any metric space $P$ with metric $d$, a ball $B(c,r)$ with center $c\in P$ and radius $r$ is defined as the set of points at a distance at most $r$ from $c$, i.e., $B(c,r)=\{p\in P\mid d(c,p) \le r\}$. In the \textit{Metric Capacitated Covering} (MCC) problem, we are given a set of balls $\cb$ in the metric space $P$ with metric $d$. We are also given a capacity parameter $U\in \mathbb{N}$ for the balls. The goal is to find a minimum sized subset ${\cb}'\subseteq \cb$ and an assignment $\phi: P\rightarrow {\cb}'$ such that for any point $p\in P$, the ball $\phi(p)$ contains $p$ and the number of points assigned to a ball $B\in {\cb}'$ via $\phi$ is at most $U$, i.e., $\vert {\phi}^{-1}(B)\vert \le U$. For $B_i\in \cb$, we denote its center and radius by $c_i$ and $r_i$, respectively.  

The greedy algorithm of \cite{Wolsey82} yields an $O(\log |P|)$-approximation for MCC. Indeed, this approximation factor is tight, which can be proved using the following simple reduction from set cover. For each element, add a point. For each set, add a ball of radius 1. If an element is in a set, then the distance between the center of the corresponding ball and the corresponding point is set to 1. Consider the metric space induced by the centers and the points. The capacity of each ball is set to the total number of elements, say $n$. Now, if there is a set cover of size $k$, then all the points can be covered by $k$ balls without violating the capacities. The converse is also true. As set cover is hard to approximate within a factor of $o(\log n)$ under standard complexity theoretic assumptions \cite{Feige98}, it is not possible to find an approximation for MCC which is asymptotically better than $O(\log n)$. 

As it is not possible to obtain a $o(\log n)$-approximation for MCC, researchers have focused on obtaining bicriteria approximation. An $(\alpha,\beta)$ bicriteria approximation for MCC is a solution where the balls can be expanded by a factor of $\beta$ (i.e., for a ball $B_i\in \cb$ and a point $p_j$ assigned to $B_i$, $d(c_i,p_j)\le \beta \cdot r_i$) and the size of the solution is at most $\alpha$ times the optimum solution size (that does not expand the balls). From the above reduction, it follows that no $(o(\log n),\beta)$ bicriteria approximation is possible for MCC under standard complexity theoretic assumptions for any $\beta < 3$. This is true, as in the construction for a ball $B_i$ that does not contain a point $p_j$, the distance between $c_i$ and $p_j$ is at least $3$. Thus, with less than $3$ factor expansion, $B_i$ cannot contain any more points than before.  

On the positive side, Bandyapadhyay~{et al.}~\cite{bandyapadhyay2019capacitated} obtained an $(O(1),6.47)$ bicriteria approximation for the problem, i.e., with only a $6.47$ factor expansion of the balls it is possible to obtain a constant approximation. Their algorithm is based on rounding of the natural LP relaxation of MCC. One problem that was left open by the work of \cite{bandyapadhyay2019capacitated} is to reduce the gap between the lower bound $3$ and the upper bound $6.47$. Thus, for what possible value of $3\le \beta < 6.47$ can one  obtain an $(O(1),\beta)$ bicriteria approximation for MCC? They also consider a generalization of MCC -- \textit{Metric Monotonic Capacitated Covering} (MMCC). This problem is similar to MCC except each ball $B_i$ has its individual capacity $U_i\in \mathbb{N}$ which must be satisfied if it is chosen in the solution and the capacities are \textit{monotonic} -- for any two balls $B_i$ and $B_j$ if the radius of $B_i$ is at least the radius of $B_j$, then $U_i \ge U_j$. At first glance, this assumption might seem artificial. However, this model has applications in wireless network. In a wireless network, coverage areas of antennas can be modelled using balls. Moreover, it might be economical to invest in capacity of an antenna to serve more clients, if its coverage area is larger. Bandyapadhyay~{et al.}~\cite{bandyapadhyay2019capacitated} gave an $(O(1),9)$ bicriteria approximation for MMCC using the same approach. 

\subsection{Our Results and Techniques}
In this paper, we obtain improved results both for MCC and MMCC. 

\begin{itemize}
    \item For MCC, we obtain an $(O(1),4.24)$ bicriteria approximation, i.e., it is possible to obtain an $O(1)$-approximation with only $4.24$ factor expansion of the balls when the capacities are uniform. 
    
    \item For MMCC, we obtain an $(O(1),5)$ bicriteria approximation, i.e., it is possible to obtain an $O(1)$-approximation with only $5$ factor expansion of the balls when the capacities are monotonic. 
\end{itemize}

Similar to \cite{bandyapadhyay2019capacitated} our results are also based on LP rounding. Indeed, our starting point is their rounding algorithm. For the purpose of giving an overview of our technique, let us focus on MMCC. The algorithm in \cite{bandyapadhyay2019capacitated} consists of three steps -- Preprocessing, Cluster Formation and Selection of Balls. Each of Preprocessing and Selection of Balls incurs an overhead of a factor $3$ expansion of the balls, resulting in the $9$ factor expansion. In our algorithm we judiciously avoid the preprocessing step to save the factor $3$ expansion. At first glance, it is not entirely clear how to do the rounding without preprocessing, as the preprocessed solution has several \say{nice} properties. Nevertheless, we partition the set of points into two subsets and construct two auxilliary LPs. Using the initial fractional LP solution, we construct two feasible fractional solutions to these two LPs. We round these two solutions independently to obtain two integral solutions corresponding to the two subsets of points. For rounding the first LP, we use an algorithm similar to the one in \cite{bandyapadhyay2019capacitated}, but without preprocessing. We show that the constructed fractional LP solution has equally nice properties so that the algorithm in \cite{bandyapadhyay2019capacitated} can be extended in this case. For rounding the second LP, we use a rather simple algorithm. 

The sets of balls involved in two LPs are not necessarily disjoint, and thus a ball can be selected in both of the solutions. But, taking multiple copies of a ball is not allowed. To resolve this issue, we first identify a subset of balls and allow only these balls to be involved in both solutions. Moreover, we scale down the capacities of these balls by a suitable factor. This ensures that even if a ball is selected in both solutions, the total capacity used by the copies does not exceed the original capacity. Note that the scaling of capacities leads to a new issue that the capacities no longer satisfy the monotonicity property in general. However, we show that it is possible to overcome this issue by considering two classes of balls separately -- one whose capacities remain unchanged and the other whose capacities are scaled down. 

\subsection{Related Work}
Considering the hardness of MCC, researchers have studied the Euclidean version of the problem with the goal of obtaining better approximation. The dimension $d$ of the space is assumed to be a constant. One interesting case is when the set $\cb$ contains all possible unit balls, which appeared in the Sloan Digital Sky Survey project \cite{LuptonMY98}. Ghasemi and Razzazi \cite{GhasemiR14} have obtained a PTAS for this case. In the general Euclidean case the best known approximation factor is still $O(\log n)$. Bandyapadhyay~{et al.}~\cite{bandyapadhyay2019capacitated} showed that in this special case of MCC only $1+\epsilon$ expansion of the balls is sufficient to obtain a constant approximation. 

MCC is a special version of \textit{Capacitated Set Cover} (CSC). CSC is similar to set cover except each set $S_i$ has a capacity $U_i$. Moreover, we want to find an assignment of the points to the chosen subfamily of sets such that each element is assigned to a set it is in and at most $U_i$ elements are assigned to each set $S_i$. CSC is a well-studied problem. Wolsey \cite{Wolsey82} designed a greedy algorithm for CSC that achieves a tight $O(\log n)$-approximation. Capacitated vertex cover is another special case of CSC, where each element is contained in exactly two sets. A $3$-approximation for this problem was given by Chuzhoy and Naor \cite{ChuzhoyN06}. The approximation factor was subsequently improved to $2$ by Gandhi~{et al.}~\cite{GandhiHKKS06}. The generalization where each element belongs to at most a bounded number of sets is also well-studied \cite{Kao17,Wong17}. 

The uncapacitated version of MCC (\textit{Metric Uncapacitated Covering} (MUC)), where each set can be assigned with any number of points is another extensively studied problem. Note that the same bicriteria hardness of MCC mentioned above holds even for MUC. But, using a simple LP rounding scheme one can obtain a $(1,3)$ bicriteria approximation for this problem. The MUC problem in the fixed-dimensional Euclidean space also has received huge attention from the researchers. Br$\ddot{\text{o}}$nnimann and Goodrich \cite{BronnimannG1995} have designed an $O(1)$-approximation for this problem in the plane. In a celebrated work, Mustafa and Ray \cite{MustafaR2010} improved this result by obtaining a PTAS for the problem. In dimension more than $2$, the problem is notoriously hard and the best known approximation is $O(\log n)$. Considering this situation Har-Peled and Lee \cite{Har-PeledL12} gave a $(1+\epsilon,1+\epsilon)$ bicriteria approximation.    

Capacitated clustering and facility location problems are another set of interesting and well-studied problems. One such interesting problem is capacitated $k$-center. $O(1)$-approximations are known both for the uniform \cite{Bar-IlanKP93,KhullerS00} and non-uniform \cite{AnBCGMS15,CyganHK12} version of this problem. Another popular clustering problem is capacitated $k$-median for which no $O(1)$-approximation is known so far. Seemingly the existing techniques are not capable of handling the combination of the global constraint on the number of centers and the capacity constraint. Indeed, if either of these constraints is allowed to be violated by an $O(1)$ factor, $O(1)$-approximations are known in those cases \cite{ByrkaRU16,ByrkaFRS15,CharikarGTS02,ChuzhoyR05,DemirciL16,Li15,Li17}. For capacitated facility location $O(1)$-approximations are known based on local search paradigm \cite{AggarwalLBGGGJ13,BansalGG12,ChudakW05,KorupoluPR00,PalTW01} and rounding of LP \cite{AnSS17}.   

\subsection{Paper Outline}
In Section \ref{sec:lp} we describe the natural LP for MMCC and have some definitions, which will be useful throughout the paper. In Section \ref{sec:oldalgo} we give an overview of the algorithm of \cite{bandyapadhyay2019capacitated}. Our LP rounding algorithm for MMCC and the analysis appear in Section \ref{sec:newalgo}. In Section \ref{sec:uniform} we show how to modify our algorithm for MMCC in the uniform case to obtain the improved bound. Finally, in Section \ref{sec:conclude} we conclude with some open problems. 

\section{Preliminaries}\label{sec:lp}
Recall that in MMCC we are given a set of points $P$ and a set of balls $\cb$. The capacity of each ball $B_i\in \cb$ is $U_i$. Also, these capacities satisfy monotonicity, i.e., for any two balls $B_i$ and $B_j$, if $r_i \ge r_j$, $U_i\ge U_j$. 

The relaxation of the natural LP for MMCC is shown in the following. In the LP for MMCC, we have a variable $y_i$ for each ball $B_i\in \cb$ that indicates whether $B_i$ is in the solution ($y_i=1$) or not ($y_i=0$). For each ball $B_i$ and each point $p_j \in P$, there is a variable $x_{ij}$ that indicates whether $p_j$ is assigned to $B_i$ ($x_{ij}=1$) or not ($x_{ij}=0$). Constraint \ref{constr:open} ensures that if a point is assigned to a ball, the ball must be selected in the solution. Constraint \ref{constr:cap} ensures that the total number of points assigned to $B_i$ is at most $U_i$. Constraint \ref{constr:flow} ensures that each point is assigned to exactly one ball. Constraint \ref{constr:coverage} ensures that if a point $p_j$ is assigned to a ball $B_i$, $p_j$ must be contained in $B_i$. The remaining constraints are relaxed in \ref{LP}, which define the domains of the variables. We note that the LP relaxation for MCC is same as \ref{LP} except there all the $U_i$ are equal.  
\begin{align}
\label{LP}
&\text{minimize}&\sum_{B_i \in \cb} y_i \nonumber \tag{MMCC-LP}
\\&\text{s.t.}& x_{ij} &\le y_i & & \forall p_j \in P,\; \forall B_i \in \cb
\label{constr:open}
\\& &\sum_{p_j \in P} x_{ij} & \le y_i \cdot U_i & & \forall B_i \in \cb \label{constr:cap}
\\& &\sum_{B_i \in \cb} x_{ij} &= 1 & & \forall p_j \in P
\label{constr:flow}
\\& &x_{ij} &= 0 &&\text{$\forall p_j \in P,\; \forall B_i \in \cb$ such that $p_j \not\in B_i$} \label{constr:coverage}
\\& &x_{ij} &\ge 0 &&\forall p_j \in P,\; \forall B_i \in \cb \label{constr:fractional_x}
\\& &0\le y_i &\le 1 &&\forall B_i \in \cb \label{constr:fractional_y}
\end{align}

We denote any solution to \ref{LP} by $(x,y)$. To distinguish between two different solutions, we use different annotations with $x$ and $y$. The cost of $(x,y)$ is defined as, $\text{cost}(x,y)=\sum_{B_i \in \cb} y_i$. For an integral solution, the cost is exactly the number of balls in the solution. Consider any solution $(x,y)$ to \ref{LP}. For a ball $B_i$ and a point $p_j$, if $x_{ij} > 0$, we say $B_i$ \textit{serves} $p_j$ and $p_j$ receives $x_{ij}$ amount of \textit{flow} from $B_i$. The \textit{flow out} of $B_i$ is the total amount of flow $\sum_{p_j \in P} x_{ij}$ that $B_i$ gives to all the points. Next, we define an operation that we call \say{reroute}. For a point $p_j$ and two balls $B_i$ and $B_{\ell}$, rerouting of $f$ amount of flow for $p_j$ from $B_i$ to $B_{\ell}$ means we increase $x_{\ell j}$ by $f$ and decrease $x_{ij}$ by $f$. For two balls $B_i$ and $B_{\ell}$, rerouting of flow from $B_i$ to $B_{\ell}$ means for each point $p_j$ served by $B_i$, we reroute $x_{ij}$ amount of flow for $p_j$ from $B_i$ to $B_{\ell}$. Thus, the flow out of $B_i$ becomes $0$ after this operation. For a point $p_j$, a set of balls $S$ and a ball $B_{\ell}\notin S$, rerouting of $f$ amount of flow from the balls in $S$ to $B_{\ell}$ means we increase $x_{\ell j}$ by $f$ and decrease $x_{ij}$ by $f_i\ge 0$ for each $B_i \in S$ such that $\sum_{B_i\in S} f_i=f$.  

\section{Overview of the Algorithm of \cite{bandyapadhyay2019capacitated}}\label{sec:oldalgo}

Our algorithm is based on the algorithm of \cite{bandyapadhyay2019capacitated}. In this section we give an overview of the algorithm of \cite{bandyapadhyay2019capacitated}. 
Let $(x,y)$ be a feasible solution to \ref{LP}. The LP rounding algorithm of \cite{bandyapadhyay2019capacitated} rounds the solution so that $y$ values of all the balls become integral. We note that it is sufficient to obtain such a solution. Indeed, as all the capacities are integral, it is possible to find another solution with the same $y$ values where all the $x$ values are also integral \cite{ChuzhoyN06}. The algorithm has three major steps. The first step is the preprocessing step. Fix a $0 <  \alpha\le 3/8$. A ball $B_i$ is called \emph{heavy} if $y_i =1$ and \emph{light} if $0\le y_i\le \alpha$. Let $\ch$ and $\cl$ be the respective set of heavy and light balls. We note that the sets of heavy and light balls are always defined w.r.t. an LP solution. But, for simplicity we do not explicitly mention that in the notations $\ch$ and $\cl$. The implicit solution w.r.t. which $\ch$ and $\cl$ are defined can be easily derived from the context.  
Now, it might not be true that for all $p_j\in P$, the sum of the $y$ values of the balls in $\cl$ that serve $p_j$ is at most $\alpha$. In the preprocessing step, the algorithm of \cite{bandyapadhyay2019capacitated} modifies the computed LP solution to obtain another LP solution such that the above mentioned property is satisfied. In particular, they prove the following lemma. 

\begin{lemma}
\label{lemma:preproc}(Lemma 3.1 of \cite{bandyapadhyay2019capacitated}) Given a feasible LP solution $\sol=(x, y)$, and a parameter $0 < \alpha \le
\frac{3}{8}$, there exists a polynomial time algorithm to obtain another LP solution
$\solbar=(\overline x, \overline y)$ that satisfies all the constraints of \ref{LP} (Constraints \ref{constr:open}-\ref{constr:fractional_y}), except Constraint \ref{constr:coverage}. Additionally, $\solbar$ satisfies the following properties.
\begin{enumerate}
   \item Any ball $B_i \in \cb$ with non-zero $\overline{y_i}$ is either heavy ($\overline{y_i} = 1$) or light ($0<
      \overline{y_i} \le \alpha$).
   \item For each point $p_j \in P$, we have that 
       \begin{align}
       \label{ineq:cond2}
          \sum_{B_i \in \cl:
         \overline{x}_{ij} > 0} \overline{y_i} \leq  \alpha,
       \end{align}
where $\cl$ is the set of light balls with respect to $\solbar$. \label{Condition 2}
   \item For any heavy ball $B_i$, and any point $p_j \in P$ served by $B_i$,
      $d(c_i, p_j) \le 3r_i$.
   \item For any light ball $B_i$, and any point $p_j \in P$ served by $B_i$, $d(c_i, p_j) \le r_i$.
   \item $\cost(\solbar) \le \frac{1}{\alpha} \cost(\sol)$.
\end{enumerate}
\end{lemma}

Note that a point $p_j$ can be fractionally assigned by the algorithm in Lemma \ref{lemma:preproc} to a heavy ball $B_i$ even if $p_i \notin B_i$, but, in this case $d(c_i, p_j)$ must be at most $3r_i$. Hence, a factor 3 expansion of the ball is sufficient for it to serve the point. In summary, the preprocessing step implicitly incurs an expansion factor of 3 for the heavy balls with respect to the new LP solution $\solbar$. We also note that the preprocessing algorithm uses the fact that the capacities are monotonic. 

The second step of the algorithm is the key step and is called Cluster Formation. In the following, we give an overview of this step. 
The algorithm maintains an LP solution $\solbar=(\overline x, \overline y)$ which is initially the output of the preprocessing step. This solution is essentially altered throughout the step and when the step finishes $\overline y_i \in \{0,1\}$ for all $B_i \in \cb$. Each heavy ball $B_i$ forms a cluster which initially consists of itself  ($\{B_i\}$). For any light ball $B_t$, either $B_t$ is opened fully in the solution or it joins a cluster of a heavy ball by rerouting its flow to the heavy ball. The algorithm runs for several iterations until the fate of all these light balls are decided. 

In each iteration, every heavy ball uses its available capacity to reroute the flow of as many intersecting light balls as possible to itself. Each such light ball joins the cluster of the heavy ball. From the remaining light balls whose fate are not yet decided, a ball is selected greedily to be included in the solution. Also, for points inside the selected ball, an appropriate amount of flow is rerouted from other balls to this ball to utilize its capacity. We skip the details of this flow rerouting in this overview. This completes the overview of the step. 

Note that the flow rerouting from heavy balls to a light ball when the light ball is opened fully, is an essential component of the analysis for obtaining the constant factor guarantee on the size of the solution. Consider a light ball $B_t$ which is selected for opening fully and assume that it serves $k_t\le U_t$ points. Then, we can set the $\overline x_{tj}$ value for each of these $k_t$ points to 1, i.e., we fully assign $p_j$ to $B_t$. Note that preprocessing ensures that $\sum_{B_i \in \cl:
\overline{x}_{ij} > 0} \overline{y_i} \leq  \alpha$ or $\sum_{B_i \in \ch:
\overline{x}_{ij} > 0} \overline{y_i} \geq  1-\alpha$. Thus, when these points are fully assigned to $B_t$, at least $(1-\alpha)k_t$ amount of flow is rerouted from the heavy balls to $B_t$ which they can now use to reroute flow from other light balls. This argument is essential in the analysis. 
Now, we have an observation which follows due to the way light balls are added to a cluster. 

\begin{observation}\label{obs:heavyclusters}
 Consider a cluster of a heavy ball $B_h$ that contains the light balls $B_1,\ldots, B_{\ell}$. Then, when the Cluster Formation finishes,
 \begin{enumerate}
     \item For each $1\le i\le {\ell}$, there is a point $p_j$ such that $p_j \in B_h\cap B_i$.
     \item $\sum_{i=1}^{\ell}\sum_{j\in P} \overline x_{ij} \le U_h-\sum_{j\in P} \overline x_{hj}$, i.e., the total amount of flow out of the balls in the cluster of $B_h$ is at most $U_h$. 
 \end{enumerate}
\end{observation}

The third step is called Selection of Balls. In this step, from each cluster a ball is carefully selected and expanded so that it can serve all the points served by the balls in the cluster. For a cluster of a heavy ball $B_h$, if it is the largest ball in the cluster then $B_h$ is selected and with three factor expansion it can serve all the points served by the cluster. As during preprocessing the heavy ball might have been expanded by a factor of 3, its total expansion factor is 9. If $B_h$ is not the largest ball, the largest ball $B_{\ell}$ is a light ball of the cluster. Then, we select this light ball and expand by a factor of 5 so that it can serve all the points served by the cluster. The light ball can serve the total flow assigned to the cluster, as $U_{\ell} \ge U_h$ due to monotonicity. This is another place where the monotonicity assumption on the capacities is necessary. 

The following lemma that states the guarantee achieved by the above algorithm follows due to the analysis of \cite{bandyapadhyay2019capacitated}. 

\begin{lemma}\label{lemma:cost}
There is a $(6+5\alpha)/\alpha$-approximation for MMCC that expands the balls by at most a factor of $9$. 
\end{lemma}


\section{The Modified Algorithm for MMCC} \label{sec:newalgo}
In this section, we describe our algorithm. Note that among the 9 factor expansion needed in the algorithm of \cite{bandyapadhyay2019capacitated} 3 factor is contributed by the preprocessing step. Our algorithm avoids this preprocessing step to save this factor 3 expansion. 

Fix $0< \alpha\le 1/60$. We first compute a fractional LP solution $\sol^*=(x^*,y^*)$ to \ref{LP}. Set $y_i=1$ if $y_i^* > \alpha$, otherwise $y_i=y_i^*$. Also, set $x=x^*$. Note that $\sol=(x,y)$ is a feasible solution to \ref{LP} such that $\text{cost}(\sol)\le \text{cost}(\sol^*)/\alpha$. We define the sets $\ch$ and $\cl$ of heavy and light balls w.r.t. $\sol$ in the same way, i.e., $\ch=\{B_i \mid y_i=1\}$ and $\cl=\{B_i \mid  0 < y_i \le \alpha\}$. Note that in $\sol$, any ball that gives some flow to a point is either a heavy or a light ball. We take one copy of the set of heavy balls and two copies of the set of light balls. Let these sets be $\ch_1$, $\cl_1$ and $\cl_2$, respectively. 

Next, we partition the point set into two subsets. Let $P_1$ be the subset of points such that $p_j \in P_1$ if $\sum_{B_i \in \cl}
{x}_{ij} \leq  4\alpha$, i.e., $p_j$ gets a flow of at most $4\alpha$ from the balls in $\cl$. Let $P_2=P\setminus P_1$. Based on these sets $P_1,P_2$, we are going to construct two LP solutions to two auxilliary LPs and round them independently. Finally, we combine these two solutions to find a solution to \ref{LP} where for each $B_i\in \cb$, $y_i\in \{0,1\}$. Intuitively, we satisfy the demands of these two sets of points independently. The light balls are involved in both of the solutions and they might get opened fully in both of the solutions. However, we are not allowed to open multiple copies of a ball. To avoid this situation we reduce the capacity of the light balls by appropriate factor in the auxilliary LP. 

Let the new capacity $U_i'=U_i/10$ for each light ball $B_i$. The new capacity of each heavy ball $B_i$ remains same as before, i.e., $U_i'=U_i$. At this point the reader might wonder about the value of the scaling factor. We note that it is carefully chosen through back calculation to ensure that the analysis goes through. 
The first auxilliary LP that we consider is as follows. 

\begin{align}
\label{LP1}
&\text{minimize}&\sum_{B_i \in \cl_1 \cup \ch_1} y_i \nonumber \tag{AUX-LP1}
\\&\text{s.t.}& x_{ij} &\le y_i & & \forall p_j \in P_1,\; \forall B_i \in \cl_1\cup\ch_1
\label{constr:open1}
\\& &\sum_{p_j \in P_1} x_{ij} & \le y_i \cdot U_i' & & \forall B_i \in \cl_1\cup\ch_1 \label{constr:cap1}
\\& &\sum_{B_i \in \cl_1\cup\ch_1} x_{ij} &= 1 & & \forall p_j \in P_1
\label{constr:flow1}
\\& &x_{ij} &= 0 &&\text{$\forall p_j \in P_1,\; \forall B_i \in \cl_1\cup\ch_1$ such that $p_j \not\in B_i$} \label{constr:coverage1}
\\& &x_{ij} &\ge 0 &&\forall p_j \in P_1,\; \forall B_i \in \cl_1\cup\ch_1 \label{constr:fractional_x1}
\\& &0\le y_i &\le 1 &&\forall B_i \in \cl_1\cup\ch_1 \label{constr:fractional_y1}
\end{align}

Note that the above LP has a variable $y_i$ for each ball $B_i$ in $\cl_1\cup\ch_1$, and a variable $x_{ij}$ for each ball $B_i$ in $\cl_1\cup\ch_1$ and each point $p_j\in P_1$. We are not going to solve this LP. Instead, we construct a solution to this LP using $\sol$ and round it using an algorithm similar to the one in \cite{bandyapadhyay2019capacitated}. This LP is used to compare the cost of the rounded solution and the cost of $\sol^*$ in the end. 

We construct an LP solution $\solbar=(\overline x,\overline y)$ from $\sol$ in the following manner. For $B_i \in \ch_1$, $\overline y_i=y_i$. For $B_i\in \cl_1$, $\overline y_i=10\cdot y_i\le 10\alpha < 1$ ($\alpha \le 1/60$). For $p_j\in P_1$, $B_i\in \cl_1\cup\ch_1$, $\overline x_{ij}=x_{ij}$. 

\begin{lemma}
$\solbar=(\overline x,\overline y)$ is a feasible solution to \ref{LP1} with cost at most $ \emph{cost}(\sol^*)/\alpha$. 
\end{lemma}

\begin{proof}
First note that, $$\text{cost}(\solbar)= \sum_{B_i\in \ch_1} y_i+10\sum_{B_i\in \cl_1} y_i\le (1/\alpha)\sum_{B_i\in \ch_1} y_i^*+10\sum_{B_i\in \cl_1} y_i^*\le \text{cost}(\sol^*)/\alpha.$$ For $p_j\in P_1$, $B_i\in \cl_1\cup\ch_1$, $\overline x_{ij}=x_{ij}\le y_i\le \overline y_i$. 
Thus, Constraint \ref{constr:open1} is satisfied.  

For $B_i \in \ch_1$, $\sum_{p_j \in P_1} \overline x_{ij}=\sum_{p_j \in P_1} x_{ij}\le y_i \cdot U_i=\overline y_i \cdot U_i'$. For $B_i\in \cl_1$, $\sum_{p_j \in P_1} \overline x_{ij}= \sum_{p_j \in P_1} x_{ij}\le y_i \cdot U_i=(10\cdot y_i) \cdot (U_i/10)=\overline y_i \cdot U_i'$. Thus, Constraint \ref{constr:cap1} is satisfied.

For $p_j\in P_1$, $\sum_{B_i \in \cl_1\cup\ch_1} \overline x_{ij}=\sum_{B_i \in \cl_1\cup\ch_1} x_{ij}=1$. 
Thus, Constraint \ref{constr:flow1} is satisfied. Also, it is trivial to verify that Constraints \ref{constr:coverage1}-\ref{constr:fractional_y1} are also satisfied. Hence, the lemma follows. 
\end{proof}

Next, we describe our second auxilliary LP. Let us again consider the solution $\sol=(x,y)$ to \ref{LP} and the set of light balls $\cl$ w.r.t. $\sol$. Also, consider the second copy  $\cl_2$ of the set of light balls. For each point $p_j$ in $P_2$, define the demand $d_j=\sum_{B_i \in \cl_2} {x}_{ij}$.  

\begin{align}
\label{LP2}
&\text{minimize}&\sum_{B_i \in \cl_2} y_i \nonumber \tag{AUX-LP2}
\\&\text{s.t.}& x_{ij} &\le y_i & & \forall p_j \in P_2,\; \forall B_i \in \cl_2
\label{constr:open2}
\\& &\sum_{p_j \in P_2} x_{ij} & \le y_i \cdot U_i' & & \forall B_i \in \cl_2 \label{constr:cap2}
\\& &\sum_{B_i \in \cl_2} x_{ij} &\ge d_j & & \forall p_j \in P_2
\label{constr:flow2}
\\& &x_{ij} &= 0 &&\text{$\forall p_j \in P_2,\; \forall B_i \in \cl_2$ such that $p_j \not\in B_i$} \label{constr:coverage2}
\\& &x_{ij} &\ge 0 &&\forall p_j \in P_2,\; \forall B_i \in \cl_2 \label{constr:fractional_x2}
\\& &0\le y_i &\le 1 &&\forall B_i \in \cl_2 \label{constr:fractional_y2}
\end{align}

Note that the above LP has a variable $y_i$ for each ball $B_i$ in $\cl_2$ and a variable $x_{ij}$ for each ball $B_i$ in $\cl_2$ and each point $p_j\in P_2$. Again we are not going to solve this LP. Instead, we construct a solution to this LP using $\sol$ and round it. This LP is used to compare the cost of the rounded solution and the cost of $\sol^*$ in the end. 

We construct an LP solution $\hat{\sol}=(\hat{x},\hat{y})$ from $\sol$ in the following manner. For $B_i\in \cl_2$, $\hat{y}_i=10\cdot y_i\le 10\alpha < 1$. For $p_j\in P_2$, $B_i\in \cl_2$, $\hat{x}_{ij}=x_{ij}$. 

\begin{lemma}
$\hat{\sol}=(\hat{x},\hat{y})$ is a feasible solution to \ref{LP2} with cost at most $10\cdot \emph{cost}(\sol^*)$. 
\end{lemma}

\begin{proof}
First note that $\text{cost}(\hat{\sol})\le 10\sum_{B_i\in \cl_2} y_i=10\sum_{B_i\in \cl_2} y_i^*\le 10\cdot \text{cost}(\sol^*)$. For $p_j\in P_2$, $B_i\in \cl_2$, $\hat{x}_{ij}=x_{ij}\le y_i < \hat{y}_i$. Thus, Constraint \ref{constr:open2} is satisfied.  

For $B_i\in \cl_2$, $\sum_{p_j \in P_2} \hat{x}_{ij}= \sum_{p_j \in P_2} x_{ij}\le y_i \cdot U_i=(10\cdot y_i) \cdot (U_i/10)=\hat{y}_i \cdot U_i'$. Thus, Constraint \ref{constr:cap2} is satisfied.

For $p_j\in P_2$, $\sum_{B_i \in \cl_2} \hat{x}_{ij}=\sum_{B_i \in \cl_2} x_{ij}=d_j$. Thus, Constraint \ref{constr:flow2} is satisfied. Also, it is trivial to verify that Constraints \ref{constr:coverage2}-\ref{constr:fractional_y2} are also satisfied. Hence, the lemma follows. 
\end{proof}

In the following, we give two algorithms for rounding the two auxilliary LPs. The rounded solution of the first LP satisfies all the constraints except the coverage constraint. The rounded solution of the second LP satisfies all the constraints except the coverage and capacity constraints. Then, we merge these two solutions to obtain a solution for \ref{LP} that does not violate any capacity constraints. 

\subsection{Rounding the First Auxilliary LP}
Note that we are given a feasible LP solution $\solbar=(\overline{x},\overline{y})$ to \ref{LP1} that has the following properties. 
\begin{enumerate}
    \item For any $B_i\in \ch_1$, $\overline{y}_i=1$.
    \item For any $B_i\in \cl_1$, $\overline{y}_i\le 10\alpha$.
    \item For any $p_j\in P_1$, $\sum_{B_i \in \ch_1} \overline {x}_{ij} \ge 1-4\alpha$. 
    \item $\text{cost}(\overline {\sol})\le \text{cost}({\sol}^*)/\alpha$.
\end{enumerate}

Note that Property (3) above states that for any point $p_j\in P_1$, the flow received by $p_j$ from the balls in $\ch_1$ is at least $1-4\alpha$. We will heavily use this property while performing the rounding. Indeed, we are going to use an algorithm similar to the one in \cite{bandyapadhyay2019capacitated} without the preprocessing step. In the algorithm of \cite{bandyapadhyay2019capacitated}, preprocessing ensures that for any point $p_j$, the sum of the $y$ values of the light balls that give non-zero flow to $p_j$ is at most $\alpha$. Note that this might not be true in our case for balls in $\cl_1$. At first glance it is not clear how to do the rounding without this assumption. However, as we show, a similar rounding scheme can be designed using the weaker assumption on the flow mentioned above. Another hurdle to adapt the algorithm of \cite{bandyapadhyay2019capacitated} is the monotonicity assumption, which might not be true in our case because of scaling of the capacities. However, we note that only light balls' capacities are scaled by a uniform constant scaling factor. Due to this fact, we show that their algorithm can be modified to handle our case. Next, we describe our rounding algorithm. 

The first step in our algorithm is Cluster Formation. In this step, for each ball $B_i\in \cl_1$, either $B_i$ is opened fully (added to a set $\co$) and flow from other balls including the balls in $\ch_1$ are rerouted to $B_i$ only for points in $B_i$. Otherwise, $B_i$ joins a cluster of a ball in $\ch_1$ to which its entire flow is rerouted. $\co$ is initialized to the empty set. 
For each ball $B_i\in \ch_1$, initialize the cluster of $B_i$, cluster$(B_i)$ to $\{B_i\}$. 
During the course of the algorithm, let $\Lambda\subseteq \cl_1$ be the set of balls which are not yet added to $\co$ or to a cluster of a ball in $\ch_1$. Throughout the algorithm, we maintain the invariant that for any point $p_j$ which is served by a ball in $\Lambda$, $p_j$ receives a flow of at least $1-4\alpha$ from the balls in $\ch_1$. Note that in the beginning of the algorithm this is true, as $\Lambda=\cl_1$. At any point, the available capacity of a ball $B_i$, $AC(B_i)=U_i'-\sum_{j\in P_1} \overline x_{ij}$. While the set $\Lambda$ is non-empty, apply the following steps. 

\begin{description}
\item While there is a ball $B_i\in \ch_1$ and $B_{i'}\in \Lambda$ such that $B_i$ intersects $B_{i'}$ and $AC(B_i)$ is at least the flow out $\sum_{j\in P_1} \overline x_{i'j}$ of $B_{i'}$, reroute the flow from $B_{i'}$ to $B_i$. Add $B_{i'}$ to cluster$(B_i)$. If $\Lambda$ becomes empty at this point, go to the Selection of Balls stage.   

\item For any ball $B_j\in \Lambda$, let $\ca_j$ be the set of points currently being served by $B_j$. Also, let $k_j=\min\{U_j',|\ca_j|\}$. We add a ball $B_t\in \Lambda$ to $\co$ such that $k_t$ is the maximum over all $k_j$ for $B_j\in \Lambda$. 

\item Next we assign points up to larger extents to $B_t$ to utilize its capacity. There are three cases. 
 
\begin{enumerate}
    \item $k_t >2$. Note that the flow out of $B_t$, $\sum_{j\in P_1} \overline x_{tj}\le 10\alpha U_t'$. Also, as $\overline x_{tj}=x_{tj}\le y_t\le \alpha$, $\sum_{j\in P_1} \overline x_{tj}\le \alpha|\ca_t|\le 10\alpha |\ca_t|$. Thus, $AC(B_t)\ge (1-10\alpha)k_t$. In this case, we arbitrarily select $\lfloor (1-10\alpha)k_t\rfloor$ points served by $B_t$ and for each of those points $p_{\ell}$, we reroute the maximum (whole) amount of flow possible from all other balls to $B_t$. Note that $p_{\ell}$ is no longer served by a ball in $\Lambda$, and thus the invariant is satisfied. 
    \item $1\le k_t\le 2$. If $U_t' \ge |\ca_t|$, then $|\ca_t|=k_t$. In this case, for each of the $k_t$ points served by $B_t$, we reroute the maximum amount of flow possible from all other balls to $B_t$. In the other case, $U_t' < |\ca_t|$. Now, $AC(B_t)\ge (1-10\alpha)U_t'\ge 1-10\alpha$. The last inequality follows, as $U_t'\ge 1$. We arbitrarily select a point $p_{\ell}$ that is being served by $B_t$ and reroute its flow from $\Lambda$ to $B_t$. Let $f$ be the amount of flow that now $p_{\ell}$ receives from $B_t$. Note that $f\le 4\alpha$. Also, $p_{\ell}$ is no longer served by a ball in $\Lambda$. Now, $AC(B_t)\ge 1-10\alpha-4\alpha=1-14\alpha$. We reroute $\min\{AC(B_t),1-f\}$ amount of flow from $\ch_1$ to $B_t$ for $p_{\ell}$. In any case, the points whose flow are routed to $B_t$ in this step are no longer served by a ball in $\Lambda$, and thus the invariant is satisfied.
    \item $0< k_t < 1$. Note that, as $|\ca_t|\ge 1$, $k_t=U_t' < 1$. Now, $AC(B_t)\ge (1-10\alpha)U_t'$. Consider any arbitrary point $p_{\ell}$ that is being served by $B_t$. First, reroute its flow from $\Lambda$ to $B_t$. $AC(B_t)\ge (1-10\alpha)U_t'-4\alpha$. Note that after this rerouting, $p_{\ell}$ is no longer served by balls in $\Lambda$, and thus the invariant is satisfied. Let $p_{\ell}$ gets a flow of $f_1$ from the balls in $\ch_1$. By the invariant we maintain, $f_1$ is at least $1-4\alpha$. Reroute $\min\{AC(B_t),f_1\}$ amount of flow of $p_{\ell}$ from the balls in $\ch_1$ to $B_t$. 
\end{enumerate}

\end{description}

When the while loop terminates each ball in $\cl_1$ is either in $\co$ or added to a cluster. For each $B_i\in \co$, we set $\overline y_i=1$ and cluster$(B_i)=\{B_i\}$. 

We note that the third case ($0< k_t < 1$) mentioned above does not occur in the context of \cite{bandyapadhyay2019capacitated}, as in their case for each ball $B_j$, both $U_j$ and $|\ca_j|$ are at least $1$. This case appears to be the bottleneck for our algorithm and leads to a larger constant of approximation as we will describe in the analysis. 

The Selection of Balls step is more interesting in our case as the monotonicity property no longer holds in general. For a cluster of a ball in $\co$, we trivially select this ball. Consider the cluster of any ball $B_h \in \ch_1$. If $B_h$ is one of the top 10 largest balls in the cluster, then select all the balls larger than $B_h$ and also $B_h$. Only $B_h$ is expanded by a factor of 3. The flow rerouted from any selected ball of $\cl_1$ to $B_h$ in the Cluster Formation step is assigned to it. Note that for the remaining balls of $\cl_1$ which are in the same cluster and not chosen, are smaller than $B_h$, and thus can be covered by a factor 3 expansion of $B_h$. The remaining flow is assigned to $B_h$. Otherwise, the top 10 largest balls are selected all of which are in $\cl_1$. The flow rerouted from any selected ball to $B_h$ in the Cluster Formation step is assigned to the ball. Now consider the remaining flow assigned to the cluster. Also consider a point $p_j$ which receives a part of this flow and not in any of the selected balls. Then, by 5 factor expansion, any selected ball can cover $p_j$. We expand each selected ball by 5 factor and the remaining flow is assigned arbitrarily to selected balls respecting their capacity. 

\subsubsection{Analysis}

Let $I$ be the number of iterations of the outermost while loop. Also, let $L_t$ be the ball of $\cl_1$ added to $\co$ at iteration $1\le t\le I$. For a ball $B_i\in \ch_1$, let $F(L_t,B_i)$ be the amount of flow rerouted from $B_i$ to $L_t$. Let $F_t=\sum_{B_i\in \ch_1} F(L_t,B_i)$. The next lemma states that when $L_t$ is added to $\co$ sufficient amount of flow is rerouted from the balls in $\ch_1$ to $L_t$ irrespective of the value of $k_t$. 

\begin{lemma}\label{lem:boundonFt}
For $1\le i\le I$, $F_t\ge k_t/60$ for $\alpha\le 1/60$. 
\end{lemma}

\begin{proof}
To compute the flow rerouted from balls in $\ch_1$ to $B_t$ we refer to the three cases mentioned in Cluster Formation. In the first case, for $\lfloor (1-10\alpha)k_t\rfloor$ points, the flow is rerouted from $\ch_1$ to $B_t$. Note that by the invariant we maintain, for each such point $p_{\ell}$, $p_{\ell}$ receives at least $1-4\alpha$ amount of flow from the balls in $\ch_1$. It follows that, at least $1-4\alpha$ amount of flow is rerouted for $p_{\ell}$ and $F_t\ge (1-4\alpha)\lfloor (1-10\alpha)k_t\rfloor\ge (14/15)\lfloor (5/6)k_t\rfloor\ge (14/15) (1/14)k_t=k_t/15\ge k_t/60$. The second inequality follows as $\alpha\le 1/60$ and the third inequality follows as $k_t> 2$. 

In the second case, using the same argument as above, the amount of flow rerouted from $\ch_1$ to $B_t$ is at least $1-14\alpha$. As $k_t \le 2$, $F_t$ is at least $(1-14\alpha)k_t/2 \ge (23/60)k_t\ge k_t/60$. The first inequality is true for $\alpha \le 1/60$. 

In the third case, again using the same argument as above, the amount of flow rerouted from $\ch_1$ to $B_t$ is at least $\min\{(1-10\alpha)k_t-4\alpha,1-4\alpha\}$. As $k_t < 1$, $1-4\alpha \ge (1-4\alpha)k_t$. Thus, $F_t\ge (1-10\alpha)k_t-4\alpha$. As $U_t \ge 1$, $k_t=U_t'\ge 1/10$, and hence  $F_t\ge (1-10\alpha)/10-4\alpha=1/10-5\alpha\ge 1/60$. The last inequality follows from the fact that $\alpha \le 1/60$. 
\end{proof}

Define the $y$-credit of a ball $B_i\in \ch_1$ as $Y(L_t,B_i)=F(L_t,B_i)/k_t$. 
At any moment during the Cluster Formation stage, define the $y$-accumulation of $B_i$ as $\Tilde{y}(B_i)=\sum_{L_t\in \co} Y(L_t,B_i)$ $-\sum_{B_i\in \cl_1\cap \text{cluster}(B_i)} \overline y_i$. The $y$-credit $Y(L_t,B_i)$ of $B_i$ can be seen as a normalized load it transfers to $L_t$. The $y$-accumulation $\Tilde{y}(B_i)$ is basically the difference between the total $y$-credit received by $B_i$ and the sum of normalized flows of the balls absorbed by $B_i$. The next lemma gives a lower bound on the available capacities of the balls in $\ch_1$, which is similar to Lemma 3.3 of \cite{bandyapadhyay2019capacitated}. 

\begin{lemma}\label{lem:ycreditlb}
Consider a  ball $B_i\in \ch_1$  and any integer $1\le t\le I$. Suppose the balls $L_1,\ldots,L_t$ have been added to $\co$ so far. Then, $AC(B_i)\ge \Tilde{y}(B_i)k_t$. 
\end{lemma}

\begin{proof}
For any ball $B_i\in \ch_1$, we prove the claim using induction on iteration number. In the base case, just after addition of $L_1$, $AC(B_i)\ge F(L_1,B_i)=Y(L_1,B_i)k_1 = \Tilde{y}(B_i)k_1$. 
Now, suppose the claim is true for any $t-1$. We show that the claim is true for $t$ as well. 

Consider the iteration $t$. Note that $AC(B_i)\ge \Tilde{y}(B_i)k_{t-1}$. Suppose a subset of balls have joined cluster of $B_i$. Let $B_p$ be the first ball joined, which serves $k$ points. To distinguish between the old and new value of $\Tilde{y}(B_i)$, we refer to the new value by $\Tilde{y}(B_i)'$.  After $B_p$'s joining to cluster of $B_i$, $\Tilde{y}(B_i)'=\Tilde{y}(B_i)-\overline y_p$. Now, the total flow out of $B_p$ is at most $\min\{\overline y_p k,\overline y_p U_p'\}=\overline y_p\min\{k,U_p'\}\le \overline y_p k_{t-1}$. Thus, $AC(B_i)\ge \Tilde{y}(B_i)k_{t-1}-\overline y_p k_{t-1}=\Tilde{y}(B_i)'k_{t-1}$. Using the same argument it can be shown that after each subsequent addition of a ball to cluster of $B_i$ the claim is true.     

In the next step, $L_t$ is added to $\co$. Let $\Tilde{y}(B_i)$ be the y-accumulation before this. After this addition, the new y-accumulation $\Tilde{y}(B_i)' =\Tilde{y}(B_i)+Y(L_t,B_i)$. If $\Tilde{y}(B_i) \le 0$, the new available capacity $A_i'\ge Y(L_t,B_i)k_t\ge \Tilde{y}(B_i)'k_t$. Otherwise, $\Tilde{y}(B_i) > 0$, the new available capacity by the induction hypothesis is, $A_i'=AC(B_i)+Y(L_t,B_i)k_t\ge \Tilde{y}(B_i)k_{t-1}+Y(L_t,B_i)k_t\ge (\Tilde{y}(B_i)+Y(L_t,B_i))k_t=\Tilde{y}(B_i)'k_t$. 
\end{proof}

The next lemma shows that for any ball $B_i\in \ch_1$, $y$-accumulation is bounded, which is similar to Lemma 3.4 of \cite{bandyapadhyay2019capacitated}. 

\begin{lemma}\label{lem:ycreditupperbounds}
At any point, for any ball $B_i\in \ch_1$, $\Tilde{y}(B_i)< 1+10\alpha$. 
\end{lemma}

Intuitively, if the $y$-accumulation of $B_i$ exceeds the bound, it must be due to selection of a ball $L_t$ in $\cl_1$. However, one can show that $B_i$ had enough available capacity to absorb the flow from $L_t$. Hence, the bound follows. 

\begin{proof}
Let $B_i\in \ch_1$ be the first ball for which $\Tilde{y}(B_i)\ge 1+10\alpha$. As $\Tilde{y}(B_i)$ increases due to addition of balls in $\Lambda$ to $\co$, let $L_t$ be the ball whose addition increases $\Tilde{y}(B_i)$ from less than $1+10\alpha$ to at least $1+10\alpha$. 
Let $\Tilde{y}(B_i)$ and $\Tilde{y}(B_i)'$ be the y-accumulation before and after addition of $L_t$. Thus, $\Tilde{y}(B_i)< 1+10\alpha$. 
Now, $\Tilde{y}(B_i)'=\Tilde{y}(B_i)+Y(L_t,B_i)\ge 1+10\alpha$. As $\Tilde{y}(B_i)'>\Tilde{y}(B_i)$, $Y(L_t,B_i) > 0$. However, by definition $Y(L_t,B_i) \le 1$. Thus, $\Tilde{y}(B_i) \ge 10\alpha$. 

Now by Lemma \ref{lem:ycreditlb}, just before addition of $L_t$, $AC(B_i)\ge \Tilde{y}(B_i)k_{t-1} \ge 10\alpha k_t$. However, total flow out of $L_t$ is at most $10\alpha k_t$, as $L_t \in \cl_1$. Thus, $L_t$ should have joined the cluster of $B_i$, which is a contradiction. Hence, $\Tilde{y}(B_i)< 1+10\alpha$. 
\end{proof}

The following lemma gives an upper bound on the number of balls of $\cl_1$ that are fully opened. 

\begin{lemma}\label{lem:boundonO}
At the end of the Cluster Formation stage, $|\co|\le 60((1+10\alpha)|\ch_1|+\sum_{B_i\in \cl_1} \overline y_i)$. 
\end{lemma}

\begin{proof}
\begin{align*}
\sum_{B_i\in \ch_1} \Tilde{y}(B_i) &= \sum_{B_i\in \ch_1} \sum_{L_t\in \co} Y(L_t,B_i)-\sum_{B_i\in \ch_1} \sum_{B_i\in \cl_1\cap \text{cluster}(B_i)} \overline y_i\\
&\ge \sum_{B_i\in \ch_1} \sum_{L_t\in \co} F(L_t,B_i)/k_t-\sum_{B_i\in \cl_1} \overline y_i\\
&= \sum_{t=1}^I F_t/k_t-\sum_{B_i\in \cl_1} \overline y_i\\
&\ge |\co|/60-\sum_{B_i\in \cl_1} \overline y_i \qquad (F_t\ge k_t/60 \text{ by Lemma } \ref{lem:boundonFt})
\end{align*}

Also, by Lemma \ref{lem:ycreditupperbounds}, $\sum_{B_i\in \ch_1} \Tilde{y}(B_i) \le (1+10\alpha)|\ch_1|$. It follows that, $|\co|\le 60((1+10\alpha)|\ch_1|+\sum_{B_i\in \cl_1} \overline y_i)$. 
\end{proof}

We obtain the following bound on the cost of the rounded solution. 

\begin{lemma}\label{lem:approxfactor}
When the algorithm terminates the total cost of the solution is at most $10|\ch_1|+|\co|\le (70+600\alpha) \emph{cost}(\sol^*)/\alpha$.
\end{lemma}

\begin{proof}
We note that from a heavy balls' cluster at most $10$ balls are selected and all the balls in $\co$ are selected. Now, by Lemma \ref{lem:boundonO},

\begin{align*}
10|\ch_1|+|\co| &\le 10|\ch_1|+60((1+10\alpha)|\ch_1|+\sum_{B_i\in \cl_1} \overline y_i)\\
&\le (70+600\alpha)(|\ch_1|+ \sum_{B_i\in \cl_1} \overline y_i)\\
&\le (70+600\alpha) \text{cost}(\sol^*)/\alpha
\end{align*}
\end{proof}

The following lemma shows that 5 factor expansion is sufficient to serve the points assigned to each cluster.

\begin{lemma}
Using factor 5 expansion of the balls the flow of any cluster can be assigned to the chosen balls without violating the capacities.  
\end{lemma}

\begin{proof}
It is clear from the algorithm that the coverage constraints are satisfied by expanding the balls by at most a factor of 5. Here we consider the capacity constraints. Note that in the first case the capacities of the selected light balls are trivially satisfied. Also, the remaining flow assigned to $B_h$ must have an amount at most $U_h$ due to the way balls are added to a cluster. Thus, its capacity constraint is satisfied. In the other case, let the total amount of flow rerouted from the selected 10 light balls to $B_h$ in Cluster Formation step be $f$. Also, let $B_{\ell}$ be the smallest radius ball among these 10 balls. Thus, the available capacity of all these balls is at least $10U_{\ell}'-f$. Note that $U_h \le U_{\ell}$, as $B_{\ell}$ is larger than $B_h$. Now, as each light balls' capacity is reduced to a factor 10 of the original capacity and the capacity of $B_h$ remains unchanged, $U_h\le 10U_{\ell}'$. Hence, the available capacity of all these 10 balls is at least $U_h-f$. As the remaining flow is at most $U_h-f$, it follows that the capacity constraints of these balls are satisfied.        
\end{proof}

We summarize our findings in the following lemma. 

\begin{lemma}\label{lem:round1}
 The solution $(\overline x,\overline y)$ satisfies all the Constraints of \ref{LP1} except Constraint \ref{constr:coverage1}. Moreover, 
 \begin{enumerate}
    \item $\overline y_i=1$ for all $B_i \in \ch_1\cup \co$ and $\overline y_i=0$ for all other balls.
     \item For any $p_j\in P_1$, $\sum_{B_i \in \ch_1\cup \co} \overline {x}_{ij} = 1$.
     \item For any point $p_j\in P_1$, if $\overline {x}_{ij}>0$, $d(c_i,p_j)\le 5\cdot r_i$.  
     \item $\emph{cost}((\overline x,\overline y))\le (70+600\alpha) \emph{cost}(\sol^*)/\alpha$. 
 \end{enumerate}
\end{lemma}

\subsection{Rounding the Second Auxilliary LP}
Note that we are given a feasible LP solution $\hat{\sol}=(\hat{x},\hat{y})$ to \ref{LP2} that has the following properties. 
\begin{enumerate}
    \item For any $B_i\in \cl_2$, $\hat{y}_i\le 10\alpha$.
    \item For any $p_j\in P_2$, $\sum_{B_i \in \cl_2} \hat{x}_{ij} \ge 4\alpha$. 
    \item For any $p_j\in P_2$ and $B_i\in \cl_2$, $\hat{x}_{ij}\le \alpha$. 
    \item $\text{cost}(\hat{\sol})\le 10\cdot \text{cost}({\sol}^*)$.
\end{enumerate}

First, we create a new solution to \ref{LP2} from $\hat{\sol}$ which has cost at most two times that of $\hat{\sol}$. We denote the new solution as well by $\hat{\sol}$. Thus, for distinction, we denote the old values by ${\hat{y}'}_i$ and ${\hat{x}'}_{ij}$. For each $y$ variable, its new value is twice the old value. Thus, $\hat{y}_i=2{\hat{y}'}_i\le 20\alpha < 1$. The last inequality follows for $\alpha \le 1/60$. And, for each $x$ variable, its new value is twice the old value. Thus, $\hat{x}_{ij}=2{\hat{x}'}_{ij}\le 2\alpha$. Note that, now, some points might receive flow of more than 1. We adjust the $\hat{x}$ values of these points so that each such point receives 1 amount of flow. We obtain the following lemma.

\begin{lemma}\label{lem:hatxprop}
There is a feasible LP solution $\hat{\sol}=(\hat{x},\hat{y})$ to \ref{LP2} that has the following properties. 
\begin{enumerate}
    \item For any $B_i\in \cl_2$, $\hat{y}_i\le 20\alpha$.
    \item For any $p_j\in P_2$, $\sum_{B_i \in \cl_2} \hat{x}_{ij} \ge 8\alpha$. 
    \item For any $p_j\in P_2$ and $B_i\in \cl_2$, $\hat{x}_{ij}\le 2\alpha$. 
    \item $\emph{cost}(\hat{\sol})\le 20\cdot \emph{cost}({\sol}^*)$.
\end{enumerate}
\end{lemma}   

\begin{proof}
First note that $\text{cost}(\hat{\sol})\le 20\cdot \text{cost}(\sol^*)$, as the values of the $y$ variables are doubled. Next, we show that $\hat{\sol}$ is feasible. 

As the $y$ variables are doubled and $\hat{x}_{ij}\le 2{\hat{x}'}_{ij}$, $\hat{x}_{ij} \le \hat{y}_i$. Thus, Constraint \ref{constr:open2} is satisfied.  

For $B_i\in \cl_2$, $\sum_{p_j \in P_2} \hat{x}_{ij}\le  \sum_{p_j \in P_2} 2{\hat{x}'}_{ij}=2\sum_{p_j \in P_2} {\hat{x}'}_{ij}\le 2{\hat{y}'}_i \cdot U_i'=\hat{y}_i \cdot U_i'$. Thus, Constraint \ref{constr:cap2} is satisfied.

As we do not decrease the $x$ variables, unless a point gets more than 1 amount of flow, Constraint \ref{constr:flow2} is also satisfied. Also, it is trivial to verify that Constraints \ref{constr:coverage2}-\ref{constr:fractional_y2} are also satisfied. 

Properties 1, 3, and 4 follows immediately. Also, Property 2 follows from the fact that previously each point received a flow of at least $4\alpha$ from the balls in $\cl_2$. Hence, the lemma follows. 
\end{proof}

We start with the fractional solution $\hat{\sol}=(\hat{x},\hat{y})$ and round it so that $\hat{y}$ becomes integral. Throughout our algorithm we modify $\hat{\sol}$ over several steps to finally obtain the desired solution. Thus whenever we refer to $\hat{\sol}$ we refer to its current value. For any $p_j\in P_2$, let $\delta_j=\sum_{B_i \in \cl_2} \hat{x}_{ij}$. Note that $\delta_j \ge 8\alpha$. Let $S$ and $\co'$ be two disjoint sets of balls which are initialized to $\cl_2$ and $\emptyset$, respectively. Throughout we also maintain that $\sum_{B_i\in S\cup \co'} \hat{x}_{ij}=\delta_j$. Note that this is true in the beginning. Our algorithm is as follows. \\\\While there is a point $p_j\in P_2$ such that $\sum_{B_i\in S} \hat{x}_{ij} > \alpha$, we do the following. 

Let $S_j$ be the set of balls in $S$ that give flow to $p_j$, i.e., $S_j$=$\{B_i \in S: \hat{x}_{ij} > 0\}$. Note that as $\sum_{B_i\in S_j} \hat{x}_{ij}=\sum_{B_i\in S} \hat{x}_{ij} > \alpha$, $\sum_{B_i\in S_j} \hat{y}_{i} \ge \sum_{B_i\in S_j} \hat{x}_{ij} > \alpha$. Find $T \subseteq S_j$ such that $\alpha \le \sum_{B_i\in T} \hat{y}_{i}\le 21\alpha$. Such a subset can always be found using a linear scan of $S_j$, as $\sum_{B_i\in S_j} \hat{y}_{i} > \alpha$ and $\hat{y}_i\le 20\alpha$ for all $B_i\in S_j$. Let $B_t$ be the largest ball in $T$. Set $\hat{y}_t=1$ and $\hat{y}_i=0$ for each $B_i\in T$. Add $B_t$ to $\co'$. Remove all balls in $T$ from $S$. Reroute the flow from all balls in $T\setminus \{B_t\}$ to $B_t$. 

\begin{lemma}
During the course of the above algorithm, the solution $\hat{\sol}$ has cost at most $20\cdot \emph{cost}(\sol^*)/\alpha$ and satisfies all the constraints of \ref{LP2} except Constraint \ref{constr:coverage2}. Moreover, for a point $p_j\in P_2$, if $\hat{x}_{ij}>0$, $d(c_i,p_j)\le 3\cdot r_i$.  
\end{lemma}

\begin{proof}
First, we prove the feasibility of $\hat{\sol}$ using induction on the iteration number. In the beginning, the claim holds. Now, consider a particular iteration. Note that the balls for which the $\hat{y}$ values are changed are in $T$ and the points for which the $\hat{x}$ values are changed are the set of points $P'$ that receive flow from a ball in $T$. It is sufficient to show that the constraints concerning these balls and points hold. Constraint \ref{constr:open2} is satisfied as for each such point $p_j$, and the ball $B_t$, $\hat{x}_{tj}\le \delta_j\le 1=\hat{y}_t$ and for a ball $B_i\in T\setminus \{B_t\}$, $\hat{x}_{ij} =0$. Now, we argue that the capacity constraint of the ball $B_t$ is satisfied. Note that in the beginning of the iteration, the total flow out of balls in $T$ to all points is at most $$\sum_{B_i\in T} \hat{y}_i\cdot U_i'\le U_t'\sum_{B_i\in T} \hat{y}_i\le U_t'\cdot 21\alpha < U_t'.$$ The first inequality follows from the fact that $B_t$ is the largest ball in $T$ and all the capacities of the balls in $\cl_1$ are scaled by the same factor. The last inequality follows, as $\alpha\le 1/60$. Now, as this total flow is served by $B_t$ the claim holds. Constraint \ref{constr:flow2} is also satisfied for all the points in $P'$, as the flow is only rerouted from a ball to $B_t$. The other constraints except \ref{constr:coverage2} are trivial to verify. 

Note that whenever we set $\hat{y}_t=1$, we also set $\hat{y}_i=0$ for each $B_i\in T\setminus\{B_t\}$. Thus for each ball $B_t$ we can charge all the balls in $T$. As $\sum_{B_i\in T} \hat{y}_i \ge \alpha$, the cost blow up is at most a factor of $1/\alpha$. Thus, the cost is at most $20\cdot \text{cost}(\sol^*)/\alpha$. 

Whenever we reassign flow from balls in $T\setminus \{B_t\}$ to $B_t$, for a point $p_j\in P_2$, it holds that if $\hat{x}_{tj}>0$, $d(c_t,p_j)\le 3\cdot r_t$. This is true, as $B_t$ is the largest ball in $T$. As we remove $B_t$ from $S$, no flow is ever rerouted again from or to $B_t$. Hence, the claim continues to hold for all points. 
\end{proof}

Now, note that when the while loop of the above algorithm terminates, it holds that for any $p_j\in P_2$, $\sum_{B_i\in S} \hat{x}_{ij} \le  \alpha$. Thus, $\sum_{B_i\in \co'} \hat{x}_{ij}\ge \delta_j-\alpha \ge 7\alpha$. 
Using this fact, we compute a solution $(x',y')$ to \ref{LP2} (that violates Constraint \ref{constr:coverage2} and Constraint \ref{constr:cap2}). For any ball $B_i$ in $\co'$, set $y'_i=1$. For any $p_j\in P_2$ and $B_i$ in $\co'$, set $x'_{ij}=\min\{(1/(7\alpha))\cdot \hat{x}_{ij},1\}$. All the other $x'$ and $y'$ values are set to zero. Note that, now, each point receives a flow of at least 1. We adjust the $x'$ values so that each point receives exactly 1 amount of flow. We obtain the following lemma. 

\begin{lemma}\label{lem:round2}
The solution $(x',y')$ satisfies all the constraints of \ref{LP2} except Constraint \ref{constr:coverage2} and Constraint \ref{constr:cap2}. Moreover, 
 \begin{enumerate}
    \item $y'_i=1$ for all $B_i \in \co'$ and $y'_i=0$ for all $B_i \notin \co'$.
     \item For any $p_j\in P_2$, $\sum_{B_i \in \co'} x'_{ij} = 1$.
     \item For any $B_i\in \co'$, $\sum_{p_j \in P_2} x'_{ij} \le (1/(7\alpha))\cdot U_i'$.
     \item For any point $p_j\in P_2$, if $x'_{ij}>0$, $d(c_i,p_j)\le 3\cdot r_i$.  
     \item $\emph{cost}((x',y'))\le 20\cdot \emph{cost}(\sol^*)/\alpha$. 
 \end{enumerate}
\end{lemma}

\subsection{Combining the Two LP solutions}

Next, we compose the two rounded solutions obtained in Lemma \ref{lem:round1} and \ref{lem:round2} to construct a solution for the original instance. In the new solution $(\Tilde{x},\Tilde{y})$ we fully open the balls in $\ch_1\cup\co\cup \co'$. Also we keep all the ${x}$ values unchanged. Note that a ball $B_i$ of $\cl_1 (=\cl_2)$ can be opened in both solutions. However, as we had changed its capacity before, the total capacity that it can use is at most $U_i'+(1/(7\alpha))\cdot U_i'\le (1+1/(7\alpha))U_i/10 < U_i$. The last inequality follows by setting $\alpha=1/60$. The total cost of the new solution is at most $(90+600\alpha) \text{cost}(\sol^*)/\alpha\le 6000\cdot  \text{cost}(\sol^*)$. Hence, we obtain the following lemma.    

\begin{lemma}\label{thm:round}
 The solution $(\Tilde{x},\Tilde{y})$ satisfies all the Constraints of \ref{LP} except Constraint \ref{constr:coverage1}. Moreover, 
 \begin{enumerate}
     \item For any point $p_j\in P_1$, if $\overline {x}_{ij}>0$, $d( c_i,p_j)\le 5\cdot r_i$.  
     \item $\emph{cost}((\Tilde{x},\Tilde{y}))\le 6000\cdot  \emph{cost}(\sol^*)$. 
 \end{enumerate}
\end{lemma}

We note that by selecting different values of the parameters throughout the algorithm one can improve the constant in the approximation factor. However, as our main goal is to show any $O(1)$-approximation we did not pursue this.

\begin{theorem}
There is an $O(1)$-approximation for MMCC by expanding the balls by a factor of at most $5$. 
\end{theorem}

\section{Uniform Capacitated Case}\label{sec:uniform}
The algorithm in the uniform case is same except the Selection of Balls step. The next lemma shows that the Selection of Balls can be performed with only $4.24$ factor expansion of the balls. 

\begin{lemma}
Using factor $4.24$ expansion of the balls the flow of any cluster can be assigned to the chosen balls without violating the capacities. 
\end{lemma}

\begin{proof}
Consider any cluster of a heavy ball $B_h \in \ch_1$. Let $c=(1+\sqrt{5})/2$. If $B_h$ is one of the top 10 largest balls in the cluster, then select all the balls larger than $B_h$ and also $B_h$. Only $B_h$ is expanded by a factor of 3. The flow rerouted from any selected ball of $\cl_1$ to $B_h$ is assigned to the selected ball. Note that for the remaining balls of $\cl_1$ which are in the same cluster and not chosen, are smaller than $B_h$ and thus can be covered by a factor 3 expansion of $B_h$. The remaining flow is assigned to $B_h$. Note that in this case the capacities of the selected light balls are trivially satisfied. Also, the remaining flow assigned to $B_h$ must have an amount at most $U_h$. Thus, the capacity constraint of $B_h$ is satisfied.

Now, suppose $B_h$ is not one of the top 10 largest balls. Let $B_{\ell}$ be the $10^{th}$ largest ball of this cluster. Also, let $r_h$ and $r_{\ell}$ be the radius of $B_h$ and $B_{\ell}$, respectively. Now, there can be two cases (i) $r_h \ge r_{\ell}/c$ or (ii) $r_h <  r_{\ell}/c$. In the first case, we select the top 9 largest balls all of which are in $\cl_1$ and also $B_h$. The flow rerouted from any selected ball (except $B_h$) to $B_h$ is assigned to the selected ball. Now consider the remaining flow assigned to the cluster. Also consider a point $p_j$ which receives a part of this flow and not in any of the balls selected from $\cl_1$. Then, by triangle inequality, the distance between $p_j$ and the center $c_h$ of $B_h$ is at most $r_h+2r_{\ell}\le r_h+2cr_h\le 4.24 r_h$. We expand $B_h$ by the factor $4.24$ and assign the remaining flow to $B_h$. Selected balls which are in $\cl_1$ are not expanded. The capacity constraints are also satisfied due to the same reason mentioned above.

In the second case, the top 10 largest balls are selected all of which are in $\cl_1$. The flow rerouted from any selected ball to $B_h$ is assigned to the selected ball. Now consider the remaining flow assigned to the cluster. Also consider a point $p_j$ which receives a part of this flow and not in any of the selected balls. Let $B_t=B(c_t,r_t)$ be a selected ball. Then, by triangle inequality, the distance between $p_j$ and $c_t$ is at most $r_t+2r_h+2r_{\ell}\le r_t+2r_{\ell}/c+2r_{\ell}\le (3+2/c)r_t\le 4.24 r_t$. The second last inequality follows, as $r_{\ell}$ is the smallest of the selected balls. We expand each selected ball by the factor $4.24$. The remaining flow is assigned arbitrarily to selected balls respecting their capacity. Let the total amount of flow rerouted from the selected 10 light balls to $B_h$ in Cluster Formation step be $f$. The total available capacity of all these balls is at least $10U_{\ell}'-f$, as  $B_{\ell}$ is the smallest radius ball among these 10 balls. Now, as the capacity of each ball of $\cl_1$ is reduced to a factor 10 of the original capacity and the capacity of $B_h$ remains unchanged, $U_h\le 10U_{\ell}'$. Hence, the available capacity of all these 10 balls is at least $U_h-f$. As the remaining flow is at most $U_h-f$, it follows that the capacity constraints of these balls are satisfied.        

\end{proof}

\begin{theorem}
There is an $O(1)$-approximation for MCC by expanding the balls by a factor of at most $4.24$. 
\end{theorem}

\section{Conclusion}\label{sec:conclude}
In this paper, we improve the expansion factor of the balls for MCC and MMCC to 4.24 and 5, respectively, in the context of obtaining constant approximation. Our approximation factor is a large constant. But, it is possible to improve this factor by setting different values of parameters in the algorithm. Note that the lower bound on the expansion factor is still 3. So, one obvious problem is to reduce the gap further. Another interesting problem is to design a true constant approximation for the Euclidean version of MCC, which does not expand the balls. We note that this problem is open even in the plane.   

Note that if the capacities are not monotonic, no $(O(1),O(1))$-approximation is known. On the other hand, the lower bound on the expansion factor even in this case is $3-\epsilon$, similar to the uniform capacity case. So, a very natural and interesting direction of research is to study this most general version of the problem.   

\bibliography{cap-covering}

\begin{thebibliography}{10}

\bibitem{AggarwalLBGGGJ13}
Ankit Aggarwal, Anand Louis, Manisha Bansal, Naveen Garg, Neelima Gupta,
  Shubham Gupta, and Surabhi Jain.
\newblock A 3-approximation algorithm for the facility location problem with
  uniform capacities.
\newblock {\em Math. Program.}, 141(1-2):527--547, 2013.

\bibitem{AnBCGMS15}
Hyung{-}Chan An, Aditya Bhaskara, Chandra Chekuri, Shalmoli Gupta, Vivek Madan,
  and Ola Svensson.
\newblock Centrality of trees for capacitated k-center.
\newblock {\em Math. Program.}, 154(1-2):29--53, 2015.

\bibitem{AnSS17}
Hyung{-}Chan An, Mohit Singh, and Ola Svensson.
\newblock Lp-based algorithms for capacitated facility location.
\newblock {\em {SIAM} J. Comput.}, 46(1):272--306, 2017.

\bibitem{bandyapadhyay2019capacitated}
Sayan Bandyapadhyay, Santanu Bhowmick, Tanmay Inamdar, and Kasturi Varadarajan.
\newblock Capacitated covering problems in geometric spaces.
\newblock {\em Discrete \& Computational Geometry}, pages 1--31, 2019.

\bibitem{BansalGG12}
Manisha Bansal, Naveen Garg, and Neelima Gupta.
\newblock A 5-approximation for capacitated facility location.
\newblock In Leah Epstein and Paolo Ferragina, editors, {\em Algorithms - {ESA}
  2012 - 20th Annual European Symposium, Ljubljana, Slovenia, September 10-12,
  2012. Proceedings}, volume 7501 of {\em Lecture Notes in Computer Science},
  pages 133--144. Springer, 2012.

\bibitem{Bar-IlanKP93}
Judit Bar-Ilan, Guy Kortsarz, and David Peleg.
\newblock How to allocate network centers.
\newblock {\em J. Algorithms}, 15(3):385--415, 1993.

\bibitem{BronnimannG1995}
Herv{\'e} Br{\"o}nnimann and Michael~T. Goodrich.
\newblock Almost optimal set covers in finite vc-dimension.
\newblock {\em Discrete {\&} Computational Geometry}, 14(4):463--479, 1995.

\bibitem{ByrkaFRS15}
Jaroslaw Byrka, Krzysztof Fleszar, Bartosz Rybicki, and Joachim Spoerhase.
\newblock Bi-factor approximation algorithms for hard capacitated
  \emph{k}-median problems.
\newblock In Piotr Indyk, editor, {\em Proceedings of the Twenty-Sixth Annual
  {ACM-SIAM} Symposium on Discrete Algorithms, {SODA} 2015, San Diego, CA, USA,
  January 4-6, 2015}, pages 722--736. {SIAM}, 2015.

\bibitem{ByrkaRU16}
Jaroslaw Byrka, Bartosz Rybicki, and Sumedha Uniyal.
\newblock An approximation algorithm for uniform capacitated k-median problem
  with 1+{\textbackslash}epsilon capacity violation.
\newblock In Quentin Louveaux and Martin Skutella, editors, {\em Integer
  Programming and Combinatorial Optimization - 18th International Conference,
  {IPCO} 2016, Li{\`{e}}ge, Belgium, June 1-3, 2016, Proceedings}, volume 9682
  of {\em Lecture Notes in Computer Science}, pages 262--274. Springer, 2016.

\bibitem{CharikarGTS02}
Moses Charikar, Sudipto Guha, {\'{E}}va Tardos, and David~B. Shmoys.
\newblock A constant-factor approximation algorithm for the k-median problem.
\newblock {\em J. Comput. Syst. Sci.}, 65(1):129--149, 2002.

\bibitem{ChudakW05}
Fabi{\'{a}}n~A. Chudak and David~P. Williamson.
\newblock Improved approximation algorithms for capacitated facility location
  problems.
\newblock {\em Math. Program.}, 102(2):207--222, 2005.

\bibitem{ChuzhoyN06}
Julia Chuzhoy and Joseph Naor.
\newblock Covering problems with hard capacities.
\newblock {\em {SIAM} J. Comput.}, 36(2):498--515, 2006.

\bibitem{ChuzhoyR05}
Julia Chuzhoy and Yuval Rabani.
\newblock Approximating k-median with non-uniform capacities.
\newblock In {\em Proceedings of the Sixteenth Annual {ACM-SIAM} Symposium on
  Discrete Algorithms, {SODA} 2005, Vancouver, British Columbia, Canada,
  January 23-25, 2005}, pages 952--958. {SIAM}, 2005.

\bibitem{CyganHK12}
Marek Cygan, MohammadTaghi Hajiaghayi, and Samir Khuller.
\newblock {LP} rounding for k-centers with non-uniform hard capacities.
\newblock In {\em FOCS}, pages 273--282, 2012.

\bibitem{DemirciL16}
H.~G{\"{o}}kalp Demirci and Shi Li.
\newblock Constant approximation for capacitated k-median with
  (1+epsilon)-capacity violation.
\newblock In {\em 43rd International Colloquium on Automata, Languages, and
  Programming, {ICALP} 2016, July 11-15, 2016, Rome, Italy}, pages 73:1--73:14,
  2016.

\bibitem{Feige98}
Uriel Feige.
\newblock A threshold of ln \emph{n} for approximating set cover.
\newblock {\em J. {ACM}}, 45(4):634--652, 1998.

\bibitem{GandhiHKKS06}
Rajiv Gandhi, Eran Halperin, Samir Khuller, Guy Kortsarz, and Srinivasan
  Aravind.
\newblock An improved approximation algorithm for vertex cover with hard
  capacities.
\newblock {\em J. Comput. Syst. Sci.}, 72(1):16--33, 2006.

\bibitem{GhasemiR14}
Taha Ghasemi and Mohammadreza Razzazi.
\newblock A {PTAS} for the cardinality constrained covering with unit balls.
\newblock {\em Theor. Comput. Sci.}, 527:50--60, 2014.

\bibitem{Har-PeledL12}
Sariel Har{-}Peled and Mira Lee.
\newblock Weighted geometric set cover problems revisited.
\newblock {\em JoCG}, 3(1):65--85, 2012.

\bibitem{Kao17}
Mong{-}Jen Kao.
\newblock Iterative partial rounding for vertex cover with hard capacities.
\newblock In {\em SODA}, pages 2638--2653, 2017.

\bibitem{KhullerS00}
Samir Khuller and Yoram~J. Sussmann.
\newblock The capacitated \emph{K}-center problem.
\newblock {\em {SIAM} J. Discrete Math.}, 13(3):403--418, 2000.

\bibitem{KorupoluPR00}
Madhukar~R. Korupolu, C.~Greg Plaxton, and Rajmohan Rajaraman.
\newblock Analysis of a local search heuristic for facility location problems.
\newblock {\em J. Algorithms}, 37(1):146--188, 2000.

\bibitem{Li15}
Shi Li.
\newblock On uniform capacitated \emph{k}-median beyond the natural {LP}
  relaxation.
\newblock In {\em Proceedings of the Twenty-Sixth Annual {ACM-SIAM} Symposium
  on Discrete Algorithms, {SODA} 2015, San Diego, CA, USA, January 4-6, 2015},
  pages 696--707, 2015.

\bibitem{Li17}
Shi Li.
\newblock On uniform capacitated \emph{k}-median beyond the natural {LP}
  relaxation.
\newblock {\em {ACM} Trans. Algorithms}, 13(2):22:1--22:18, 2017.

\bibitem{LuptonMY98}
Robert Lupton, F.~Miller Maley, and Neal~E. Young.
\newblock Data collection for the sloan digital sky survey - {A} network-flow
  heuristic.
\newblock {\em J. Algorithms}, 27(2):339--356, 1998.

\bibitem{MustafaR2010}
Nabil~H. Mustafa and Saurabh Ray.
\newblock Improved results on geometric hitting set problems.
\newblock {\em {Discrete \& Computational Geometry}}, 44(4):883--895, 2010.

\bibitem{PalTW01}
Martin P{\'{a}}l, {\'{E}}va Tardos, and Tom Wexler.
\newblock Facility location with nonuniform hard capacities.
\newblock In {\em 42nd Annual Symposium on Foundations of Computer Science,
  {FOCS} 2001, 14-17 October 2001, Las Vegas, Nevada, {USA}}, pages 329--338.
  {IEEE} Computer Society, 2001.

\bibitem{Wolsey82}
Laurence~A. Wolsey.
\newblock An analysis of the greedy algorithm for the submodular set covering
  problem.
\newblock {\em Combinatorica}, 2(4):385--393, 1982.

\bibitem{Wong17}
Sam~Chiu{-}wai Wong.
\newblock Tight algorithms for vertex cover with hard capacities on multigraphs
  and hypergraphs.
\newblock In {\em SODA}, pages 2626--2637, 2017.

\end{thebibliography}

\end{document}